\documentclass{article}[11pt]
 
\usepackage{microtype}
\usepackage{amsmath}
\usepackage{amsfonts}
\usepackage{algorithm2e}
\usepackage[noend]{algorithmic}
\usepackage{graphicx}
\usepackage{authblk}
\usepackage{subfigure}
\usepackage[margin=1in]{geometry}

\newcommand{\qed}{\hfill \ensuremath{\Box}}

\newtheorem{lemmx}{Lemma}
\newtheorem{thmx}{Theorem}

\newtheorem{corollary}{Corollary}
\newenvironment{proof}{{\bf Proof:}}{\qed}

\SetKw{Kw}{function}

\begin{document}

\title{Deterministic algorithms for skewed matrix products}

\author{Konstantin Kutzkov}
\affil{IT University Copenhagen\\
  Denmark}
\affil{kutzkov@gmail.com}
\date{}

\maketitle

\begin{abstract}

Recently, Pagh presented a randomized approximation algorithm for the multiplication of real-valued matrices building upon work for detecting the most frequent items in data streams. We continue this line of research and present new {\em deterministic} matrix multiplication algorithms. 

Motivated by applications in data mining, we first consider the case of real-valued, nonnegative $n$-by-$n$ input matrices $A$ and $B$, and show how to obtain a deterministic approximation of the weights of individual entries, as well as the entrywise $p$-norm, of the product $AB$. The algorithm is simple, space efficient and runs in one pass over the input matrices. For a user defined $b \in (0, n^2)$ the algorithm runs in time $O(nb + n\cdot\text{Sort}(n))$ and space $O(n + b)$ and returns an approximation of the entries of $AB$ within an additive factor of $\|AB\|_{E1}/b$, where $\|C\|_{E1} = \sum_{i, j} |C_{ij}|$ is the entrywise 1-norm of a matrix $C$ and $\text{Sort}(n)$ is the time required to sort $n$ real numbers in linear space. Building upon a result by Berinde et al. we show that for skewed matrix products (a common situation in many
real-life applications) the algorithm is more efficient and achieves better approximation guarantees than previously known randomized algorithms.

When the input matrices are not restricted to nonnegative entries, we present a new deterministic group testing algorithm detecting nonzero entries in the matrix product with large absolute value. The algorithm is clearly outperformed by randomized matrix multiplication algorithms, but as a byproduct we obtain the first $O(n^{2 + \varepsilon})$-time deterministic algorithm for matrix products with $O(\sqrt{n})$ nonzero entries.     
\end{abstract}
\section{Introduction}

The complexity of matrix multiplication is one of the fundamental problems in theoretical computer science. Since Strassen's sub-cubic algorithm for matrix multiplication over a ring from the late 1960's~\cite{strassen_mm}, the topic has received considerable attention, see~\cite{mm_history} for a historical overview on the subject. It is conjectured that matrix multiplication admits an algorithm running in time $O(n^{2 + \varepsilon})$ for any $\varepsilon > 0$. For more than 20 years the record holder was the algorithm by Coppersmith and Winograd~\cite{copwin_mm} running in time $O(n^{2.376})$. Recently two results improving on~\cite{copwin_mm} were announced. In his PhD thesis Stothers~\cite{stothers_mm} presents a refinement of the Coppersmith-Winograd algorithm running in time $O(n^{2.3737})$ and Vassilevska Williams~\cite{virgi_mm} developes a general framework for obtaining a tighter upper bound on the complexity of the Coppersmith-Winograd algorithm. The latter yields the best known bound of $O(n^{2.3727})$.  Several algorithms computing exactly the matrix product for special classes have been designed. For example, algorithms with running time better than  $O(n^{2.3727})$ are known for Boolean matrix multiplication with sparse output~\cite{lingas_mm} or for the case when the input or output matrices are sparse~\cite{amossen_pagh_mm,yuster_zwick_mm}. In a recent work Iwen and Spencer~\cite{iwen_spencer_mm} present a new class of matrices whose product can be computed in time $O(n^{2+\varepsilon})$ by a deterministic algorithm: namely when the output matrix is guaranteed to contain at most $n^{0.29462}$ non-zero entries in each column (or by symmetry row). All improved algorithms use as a black-box the algebraic matrix multiplication algorithm which, unfortunately, is only of theoretical importance. It uses sophisticated algebraic approaches resulting in large constants hidden in the big-Oh notation and does not admit an efficient implementation. This motivates the need of simple  ``combinatorial-like" algorithms. 
\\\\
{\bf Approximate matrix multiplication.}
The first approximation algorithm with rigorously understood complexity by Cohen and Lewis is based on sampling~\cite{cohen_lewis_mm}. For input matrices with nonnegative entries they show a concentration around the estimate for individual entries in the product matrix with high probability. 

The amount of data to be handled has been growing at a faster rate than the available memory in modern computers. Algorithms for massive data sets, where only sequential access to the input is allowed, have become a major research topic in computer science in the last decade. Drineas et al.~\cite{drineas_et_al_mm} first recognized the need for memory-efficient methods for matrix multiplication when access to single columns and rows of the input matrices is possible. They present a randomized algorithm for approximating the product of two matrices based on sampling. The complexity of the algorithm as well as its accuracy depend on user-defined parameters. The algorithm is ``pass-efficient" since columns and rows of the input matrices are sequentially loaded into memory. The obtained approximation guarantees are expressed in terms of the Frobenius norm of the input matrices and the user-defined parameters but their method does not result in a strong guarantee for individual entries in the output matrix. Sarl\'os~\cite{sarlos_mm} observed that instead of sampling rows and columns one can use random projections to obtain a sketch of the matrix product. A notable difference to~\cite{drineas_et_al_mm} is that by sketching one obtains an additive error for each individual entry depending on the 2-norm of the corresponding row and column vector in the input matrices. 

Recently Pagh~\cite{pagh_mm} introduced a new randomized approximation algorithm. Instead of sketching the input matrices and then multiplying the resulting smaller matrices, we treat the product as a stream of outer products and sketch each outer product. Using Fast Fourier Transformation in a clever way, Pagh shows how to efficiently adapt the {\sc Count-Sketch} algorithm~\cite{count_sketch} to an outer product. The algorithm runs in one pass over the input matrices and provides approximation guarantees in terms of the Frobenius norm of their product.
\\\\
{\bf Our contribution.}
\begin{itemize}

\item A new algorithm for the case where the input matrices consist of nonnegative entries only. 
This is the first nontrivial deterministic approximation algorithm for the multiplication of nonnegative matrices in a streaming setting. Motivated by practical applications, we analyze the approximation guarantee and the algorithm complexity under the assumption that the entries adhere to Zipfian distribution. We compare it to previously known randomized algorithms and show that for certain natural settings it is more efficient and achieves better approximation guarantees. 

\item We present a new matrix multiplication algorithm for arbitrary real-valued input matrices by adapting the group testing algorithm for streams with updates in the turnstile model outlined in~\cite{muthu_survey} for detecting the entries with large absolute value in matrix product. As a byproduct we obtain the first deterministic algorithm running in $O(n^{2+\varepsilon})$ steps for matrix products with $O(\sqrt{n})$ nonzero entries.

\end{itemize}

Note that our algorithms easily generalize to rectangular matrix multiplication but for the ease of presentation we consider the case of square input matrices. Also, we will state the time complexity of our first algorithm using a function $\text{Sort}(n)$ denoting the running time of a linear space deterministic sorting algorithm. Clearly, $\text{Sort}(n) = O(n \log n)$ for comparison based sorting but under some assumptions on the elements to be sorted also better bounds are known, e.g. the $O(n \log \log n)$ time integer sorting algorithm by Han~\cite{best_sort}. 


%
%
%
%
%
%
%
\section{Preliminaries}

\subsection{Definitions}

{\bf Linear algebra.} Let $\mathbb{R_+}$ denote the field of nonnegative real numbers. Given matrices $A, B \in \mathbb{R_+}^{n\times n}$ we denote their product by $C:=AB$. The $i$th row of a matrix $A$ is written as $A_{i,*}$, the $j$th column as $A_{*,j}$. We use the term {\em entry} to identify a position in the matrix, not its value.  Thus, {\em the weight} of the entry $(i, j)$ in $A$ is the value in the $i$th row and $j$th column, $A_{ij}$, $i,j \in [n]$, for $[n]:= \{0,1,\ldots,n-1\}$. When clear from the context however, we will omit weight. 
For example, nonzero entries will refer to entries whose weight is different from 0 and by heavy entries we mean entries with large weight.

The {\em outer product} of a column vector $u \in \mathbb{R_+}^n$ and a row vector $v \in \mathbb{R_+}^n$ is a matrix $uv \in \mathbb{R_+}^{n \times n}$ such that $uv_{i,j} = u_iv_j$, $i,j \in [n]$. 
The {\em rank}  of a positive real number $a \in \mathbb{R_+}$ in a matrix $A$, denoted as  $r_A(a)$, is the number of entries strictly smaller than $a$, plus 1. Note that $a$ does not need to be present in $A$. 

The {\em $p$-norm} of a vector $u \in \mathbb{R}^n$ is $\|u\|_p = (\sum_{i=1}^n |u_i|^p)^{\frac{1}{p}}$ for $p>0$.
Similarly, we define the {\em entrywise $p$-norm} of a matrix $A \in \mathbb{R}^{n\times n}$ as $\|A\|_{Ep}:=(\sum_{i,j \in [n]}|A_{i,j}|^p)^{1/p}$ for $p \in \mathbb{N}$. The case $p=2$ is the {\em Frobenius norm} of $A$ denoted as $\|A\|_F$. The {\em $k$-residual entrywise $p$-norm} $\|A\|_{{E}^kp}$ is the entrywise $p$-norm of the matrix obtained from $A$ after replacing the $k$ entries with the largest absolute values in $A$, ties resolved arbitrarily, with 0. 
\\\\
{\bf Data streaming.}
Our algorithms have strong connection to data streaming, therefore we will use the respective terminology. 
A  {\em stream} $S$ is a sequence of $N$ updates $(i, v)$ for items $i \in \mathcal{I}$ and $v \in \mathbb{R}$. We assume $\mathcal{I}= [n]$. The {\em frequency} of $i$ is $f_i = \sum_{(i, v) \in S}v$ and ${\bf f}_S = (f_0,\ldots,f_{n-1})$ is the {\em frequency vector} of the stream $S$. The {\em insert-only} model assumes $v>0$ for all updates and in the {\em non-strict turnstile} model there are no restrictions on $v$ and the values in ${\bf f}_S$ \cite{muthu_survey}. Similarly to matrix entries, we will also refer to the frequency of an item $i$ as the {\em weight} of $i$. Items with weight above ${||f||_1}/{b}$, for a user-defined $b$, will be called {\em $b$-heavy hitters} or just {\em heavy hitters} when $b$ is clear from the context. Ordering the items in $S$ according to their absolute weight, the heaviest $b$ items in $S$ are called the {\em top-b entries} in $S$.   
\\\\
{\bf Skewed distributions.}
A common formalization of the skewness in real-life datasets is the assumption of {\em Zipfian distribution}. 
The elements in a given set $M$ over $N$ elements with positive weights follow {\em Zipfian distribution} with parameter $z>0$ if the weight of the element of rank $i$ is $\frac{|M|}{\zeta(z)i^z}$ where $\zeta(z)= \sum_{i=1}^N{i^{-z}}$ and $|M|$ denotes the total weight of elements in $M$. We will analyze only the case when the skew in the data is not light and $z > 1$.   
For $z>1$, $\sum_{i=1}^N {i^{-z}}$ converges to a small constant. We will also use the facts that for $z>1$, $\sum_{i=b+1}^N {i^{-z}} = O(b^{1-z})$ and for $z>1/2$, $\sum_{i=b+1}^N {i^{-2z}} = O(b^{1-2z})$. 

\subsection{The column row method and memory efficient matrix multiplication}

The na\"ive algorithm for the multiplication of input matrices $A, B \in \mathbb{R}^{n\times n}$ works by computing the inner product of $A_{i,*}$ and $B_{*, j}$ in order to obtain $AB_{ij}$ for all $i, j \in [n]$. An alternative view of the approach is the {\em column row method} computing the sum of outer products $\sum_{i\in [n]}A_{*,i}B_{i, *}$. While this approach does not yield a better running time, it turns out to admit algorithmic modifications resulting in more efficient algorithms. Schnorr and Subramanian~\cite{outer_prod_mm} and Lingas~\cite{lingas_mm} build upon the approach and obtain faster algorithms for Boolean matrix multiplication. Assuming that $A$ is stored in column-major order and $B$ in row-major order, the approach yields a memory efficient algorithm since the matrix product $AB$ can be computed in a single scan over the input matrices. Recently, Pagh~\cite{pagh_mm} presented a new randomized algorithm combining the approach with frequent items mining algorithms~\cite{ams,count_sketch}. Inspired by this, we present another approach to modify the column-row method building upon ideas from deterministic frequent items mining algorithms~\cite{hot,demaine_et_al,karp_et_al,misra_gries}.
\subsection{Skewed matrix products} \label{skew_appl}

In~\cite{pagh_mm} Pagh discusses several applications where one is interested only in the entries with the largest absolute value in the matrix product. We argue that the restriction the input matrices to be nonnegative is feasible for many real-life problems. Since our first algorithm is both simple and efficient, we thus believe that it will have practical value.

In~\cite{cohen_lewis_mm} the authors discuss applications of nonnegative matrix multiplication for pattern recognition problems. The presented approximation algorithm is based on sampling and is shown to efficiently detect pairs of highly similar vectors in large databases for the case where the majority of vector pairs have low similarity. We present two concrete applications of nonnegative matrix multiplication where the entries in the output matrix adhere to a skewed distribution. 

A major problem in data mining is the detection of pairs of associations among items in {\em transactional databases}, see Chapter~5 in~\cite{dm_book} for an overview. In this problem, we are given a database of $m$ transactions, where a transaction is a subset of the ground set of items  $\mathcal{I}$. We are interested in items $i, j \in \mathcal{I}$ such that an occurrence of $i$ in a given transaction implies with high probability an occurrence of $j$ in the transaction, and vice versa. The {\em lift similarity measure}~\cite{brin_et_al} formalizes the above intuition. Let $S_i$ and $S_j$ be the set of transactions containing the items $i$ and $j$, respectively. Then the lift similarity between $i$ and $j$ is defined as $\frac{|S_i \cap S_j|}{|S_i||S_j|}$. The database is usually huge and does not fit in memory, therefore one is interested in solutions performing only a few scans of the database. There are randomized algorithms for the problem~\cite{bisam,cohen_et_al} but to the best of our knowledge there is no known deterministic algorithm improving upon the straightforward solution. 

The problem can be reduced to matrix multiplication as follows. In a first pass over the database we compute the exact number of occurrences $f_i$ for each item $i$. Then we create an $n \times m$ matrix $A$ such that $A_{ij} = 1/f_i$ if and only if the item $i$ occurs in the $j$th transaction. Now it is easy to see that the weight of the entry $(i, j)$ in the product $AA^T$ is equal to lift similarity between $i$ and $j$. The nonzero entries in the $j$th column in $A$ correspond to the items in the $j$th transaction, thus the column row approach is equivalent to computing the stream of Cartesian products for all transactions such that each occurrence of the the items pair $(i, \ell)$ has weight $1/(f_if_\ell)$

Figure~\ref{fig:lift} presents the distribution of lift similarities among pairs for two real datasets: pumsb, containing census data for population and housing,  and accidents, built from data about traffic accidents in Belgium. Both datasets are publicly available at the Frequent Itemset Mining Dataset Repository (http://fimi.ua.ac.be/data/). Note that the similarities are plotted on a doubly logarithmic scale and except for a few highly correlated pairs, the similarities among item pairs decrease almost linearly. Therefore, the assumption for Zipfian distribution is realistic.

As a second example consider the popular Markov Clustering algorithm (MCL) for the discovery of communities within networks~\cite{mcl}. In a nutshell, for a given graph the algorithm first creates a column stochastic matrix $M$ based on the graph structure. The algorithm then works iteratively. In each iteration we update $M$ to the product $MM$ and then set small entries to 0. A significant proportion of the entries in each column is small. $M$ converges fast to a sparse matrix indicating clusters around certain vertices in the original graph. The most time consuming steps are the first few iterations, thus an algorithm for the detection of the significant entries in matrix products can improve the time complexity. In an additional pass one can compute the exact weight of the significant entries. 

\begin{figure}[ht]
\begin{center}

\includegraphics[scale=0.55]{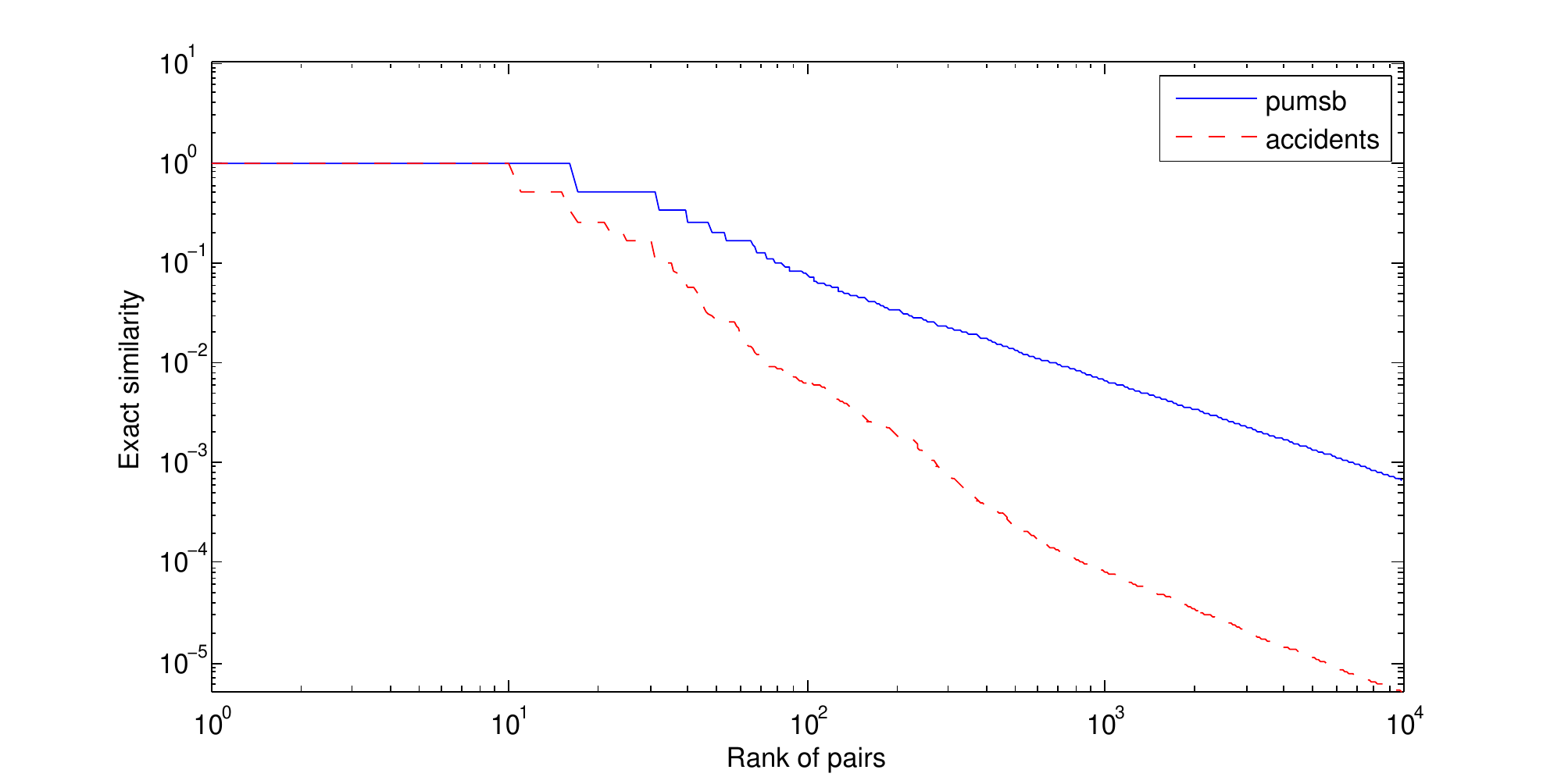}

\end{center}
\caption{The similarity distribution among pairs  for two datasets on a log-log scale. }\label{fig:lift}
\end{figure}

\section{An algorithm for nonnegative matrix products}

\subsection{Intuition and key lemma.}
Recall first how the {\sc Majority} algorithm~\cite{mjrty} works. We are given a multiset $M$ of cardinality $N$ and want to find a majority element, i.e. an element occurring at least $N/2 + 1$ times in $M$. While there are two distinct objects in $M$ we remove them from $M$. It is easy to see that if there exists a majority element $a$, at the end only occurrences of $a$ will be in $M$. 

The {\sc Frequent} algorithm~\cite{demaine_et_al,karp_et_al,misra_gries} builds upon this simple idea and detects $b$-heavy hitters in an unweighted stream $S$ of  $N$ updates $(i, 1)$ for items $i \in [n]$. We keep a \emph{summary} of $b$ distinct entries together with a counter lower bounding their weight. Whenever a new item $i$ arrives we check whether it is already in the summary and, if so, update the corresponding counter. Otherwise, if there is an empty slot in the summary we insert $i$ with a counter set to 1. In the case all $b$ slots are occupied we decrease the weight of all items by 1 and proceed with the next item in the stream. The last step corresponds to removing $b+1$ distinct items from the multiset of items occurring in $S$ and a simple argument shows that $b$-heavy hitters will be in the summary after processing the stream. By returning the estimated weight of the item in the summary and 0 for not recorded items, the weight of each item is underestimated by at most ${\|\bf{f}\|_1}/{b}$ where $\bf{f}$ is the frequency vector of the stream. Implementing the summary as a hash table and charging the cost of each item deletion to the cost incurred at its arrival the expected amortized cost per item update is constant. A sophisticated approach for decreasing the items weights in the summary leads to a worst case constant time per item update~\cite{demaine_et_al,karp_et_al}. 

Generalizing to nonnegative matrix multiplication by the column row method is intuitive. Assume the input matrices consist of \{0,1\}-valued entries only. We successively generate the $n$ outer products and run the {\sc Frequent} algorithm on the resulting stream associating entries with items. 
There are several problems to resolve: First, we want to multiply arbitrary nonnegative matrices, thus our algorithm has to handle weighted updates.
Second, we have to consider $\Theta(n^3)$ occurrences of weighted items in the stream. Third, we cannot apply any more the amortized argument for the running time analysis since a group of $b-1$ heavy items might be followed by many lighter items causing expensive updates of the summary and it is not obvious how to extend the deterministic approach from~\cite{demaine_et_al,karp_et_al} guaranteeing constant time updates in the worst case.

The first issue is easily resolved by the following

\begin{lemmx} \label{key_lemma}

Let $\bf{f}$ be the frequency vector of an insert only stream $S$ over a domain $[n]$. After successively decrementing $t$ times the weight of at least $b$ distinct items by $\Delta_i >0$, $1 \leq i \leq t$, such that at each step $f_i\geq 0$ for all $0 \leq i \leq n-1$, it holds $f_k >0$ for all $b$-heavy hitters, $k \in [n]$, for all $t \in \mathbb{N}$.


\end{lemmx}

\begin{proof} Since $f_i\geq 0$ for all $i \in [n]$ holds, the total decrease is bounded by $\|\bf{f}\|_1$. A decrement of $\Delta_i$ in the weight of a given item is witnessed by the same decrement in the weights of at least $b-1$ different items. Thus, we have $b \sum_{i=1}^t \Delta_i \leq \|\bf{f}\|_1$ which bounds the possible decrease in the weight of a heavy hitter to ${\|\bf{f}\|_1}/{b}$. 
\end{proof}

In the next section we show that the specific structure of an outer product allows us to design efficient algorithms resolving the last two issues.

\subsection{The algorithm.}

\begin{figure}[t]
\Kw{\sc ComputeSummary}
\algsetup{indent=2em}
\begin{algorithmic}[1]
\REQUIRE matrices $A, B \in \mathbb{R_+}^{n\times n}$, summary $\mathcal{S}$ for $b$ entries
\FOR{$i \in [n]$}
\STATE Denote by $R:=A_{*,i}\cdot B_{i,*}$ the outer product of the $i$th column of $A$ and $i$th row of $B$
\STATE Find the weight $w_{b+1}^R$ of the entry of rank $b+1$ in  $R$
\STATE Let $\mathcal{L}$ be the $b$ entries in $R$ with rank less than $b+1$, i.e. the largest $b$ entries
\STATE Decrease the weight of each entry in $\mathcal{L}$ by $w_{b+1}^R$
\FOR{each entry $e$ entry occurring in $\mathcal{S}$ \emph{and} $\mathcal{L}$}
\STATE add $e$'s weight in $\mathcal{L}$ to $e$'s weight in $\mathcal{S}$
\STATE remove $e$ from $\mathcal{L}$
\ENDFOR 
\STATE Find the weight $w_{b+1}^{\mathcal{S} \cup \mathcal{L}}$ of the entry of rank $b+1$ in $\mathcal{S} \cup \mathcal{L}$, if any
\STATE Update $\mathcal{S}$ to contain the largest $b$ entries in $\mathcal{S} \cup \mathcal{L}$ and decrease their weight by $w_{b+1}^{\mathcal{S}\cup \mathcal{L}}$
\ENDFOR
\end{algorithmic}
%
%
%
%
%
%
\bigskip
\algsetup{indent=2em}
\Kw{\sc EstimateEntry}
\begin{algorithmic}[1]
\REQUIRE Entry $(i, j)$
\IF{ $(i, j)$ is in the summary $\mathcal{S}$} 
\STATE return the weight of $(i, j)$ in $\mathcal{S}$
\ELSE 
\STATE return 0
\ENDIF
\end{algorithmic}
\caption{A high-level pseudocode description of the algorithm. In {\sc ComputeSummary} we iterate over the $n$ outer products and to each one of them apply Lemma~\ref{key_lemma} such that only the $b$ heaviest entries remain. We update the summary with the entries output by the outer product. After processing the input matrices we can estimate the weight of an individual entry by checking the summary.}  \label{main_alg} 
\end{figure}
We assume that $A \in \mathbb{R}_+^{n \times n}$ is stored in column-major order and $B \in \mathbb{R}_+^{n \times n}$ in row-major order. We show how to modify the column row method in order to obtain an additive approximation of each entry in $AB$ in terms of the entrywise 1-norm of $AB$. 

Essentially, we run the {\sc Frequent} algorithm for the stream of $n$ outer products: we keep a summary $\mathcal{S}$ of $b$ distinct items and for each outer product we want to update the summary with the incoming weighted entries over the domain $[n]\times [n]$. The main difference is that for $b = o(n^2)$ we can use the specific structure of an outer product and update the summary in $o(n^2)$ steps. In {\sc ComputeSummary} in Figure~\ref{main_alg} for each of the $n$ outer products we simulate the successive application of Lemma~\ref{key_lemma} until at most $b$ entries with weight larger than 0 remain in the outer product. We then update $\mathcal{S}$ with the remaining entries. 
\\\\
{\bf Correctness.}

\begin{lemmx} \label{correctness}
Let $w$ be the weight of an entry $(i,j)$ in the product $C=AB$.
After termination of {\sc ComputeSummary} for the estimated weight $\overline{w}$ of $w$ returned by {\sc EstimateEntry}, $i,j \in [n]$, holds $\mbox{max}(w- {\|C\|_{E_1}}/{b}, 0) \leq \overline{w} \leq w$.
\end{lemmx}
\begin{proof}
The product $AB$ equals $\sum_{i=0}^{n-1}a_i \cdot b_i$ for the columns $a_0,\ldots,a_{n-1}$ of $A$ and the rows $b_0,\ldots,b_{n-1}$ of $B$. We consider each outer product as $n^2$ updates for different entries over the domain $[n]\times[n]$ in an insert only stream with positive real weights. We show how the algorithm updates the summary for a single outer product $R$. First, in line 3 the algorithm finds the entry of rank $b+1$ in $R$. In line 4 we decrease the weight of the $b$ largest entries by $w_{b+1}^R$ which yields the same result as the following iterative procedure: While there are at least $b+1$ nonzero entries in $R$, find the entry with smallest weight $w_{\text{min}}$ in $R$ and decrease the weight of all non-zero entries by $w_{\text{min}}$. Equivalence holds because we always decrease the weight of an entry with the smallest weight and thus the decrease of the largest $b$ entries weights can never exceed $w_{b+1}^R$. Also, the decrease can not be smaller than $w_{b+1}^R$ since otherwise we would have more than $b$ non-zero entries in the outer product.  Thus, we always decrease by the same amount the weight of at least $b+1$ different entries which by Lemma~\ref{key_lemma} guarantees the claimed approximation error. In lines 6--10 we apply essentially the same procedure again for the nonzero entries in the outer product and the entries in the summary. The remaining at most $b$ nonzero entries constitute the updated summary. 
\end{proof}
\paragraph{Running time.} In the following lemmas we present efficient deterministic algorithms for the subroutines used in {\sc ComputeSummary}. 
We concentrate how the algorithm updates the summary for a single outer product. Before presenting our approach, we give the main building blocks that will be used to achieve an efficient solution.  

\begin{lemmx} \label{median}
Given two sorted vectors $u, v \in \mathbb{R_+}^n$  we can find the entry of rank $b$ in the outer product $uv$ in time and space $O(n)$.  
\end{lemmx}
\begin{proof}
We reduce the problem to selection of the element of rank $b$ in a Cartesian sum $X + Y = \{x+y: x \in X, y \in Y\}$ for sorted sets of real numbers $X$ and $Y$. Setting $U = \{ \log u_i: u_i \in u\}$ and $V = \{ \log v_i: v_i \in v\}$ and searching in the Cartesian sum $U + V$ for the element of rank $b$ corresponds to searching for the entry of rank $b$ in the outer product $uv$, this follows from monotonicity of the $\log: \mathbb{R_+} \rightarrow \mathbb{R}$ function. The best known deterministic algorithm for selection in a Cartesian sum of two sorted sets~\cite{selection_XY} runs in time and space $O(n)$.  
\end{proof}

\begin{lemmx} \label{rank} Given vectors $u, v \in \mathbb{R_+}^n$, with elements sorted in descending order, we can output an implicit representation of the largest $b$ elements from the outer product $uv$ in a data structure $\mathcal{L}$ in time and space $O(n)$. 
\end{lemmx} 
\begin{proof} 
Assume we have found the entry of rank $b$ as outlined in Lemma~\ref{median}, let this element be $c$.
Let $i$, $j$ be two pointers for $u$ and $v$ respectively. Initialize $i=0, j=n-1$. Assume $i$ is fixed. We compare $c$ to $u_i v_j$. While it is larger or equal, we move left $v$'s pointer by decreasing $j$ by 1. At the end we add the pair $(i, j)$ to $\mathcal{L}$, denoting that the entries in $i$th row of $uv$ bigger than $c$, and thus of rank less than $b$, are all $(i, \ell)$ for $\ell \le j$. Then we go to the next row in $uv$ by incrementing $i$ and repeat the above while-loop starting with the current value of $j$. When $i=n$ or $j=0$ we know that all entries smaller than $c$ have been found. Correctness is immediate since the product $u_iv_j$ is monotonically increasing with $i$ and decreasing with $j$, and thus for each row of the outer product we record the position of the entries smaller than $c$ in $\mathcal{L}$.  Both $i$ and $j$ are always incremented or respectively decremented, thus the running time is linear in $n$. We need to explicitly store only $u, v$ and $\mathcal{L}$, this gives the claimed space usage. 
\end{proof}

Next we present an efficient approach for updating the summary for a given outer product after finding the entries of rank at most $b$.

\begin{lemmx} \label{rtime_op}
For a given outer product $uv$, $u, v \in \mathbb{R_+}^n$ we update the summary $\mathcal{S}$ in time $O(b + \text{sort}(n))$.
\end{lemmx}
\begin{proof}

We first sort the vectors $u$ and $v$ in decreasing order according to the values $u_i$ and $v_i$, respectively. Let us call the sorted vectors $u^s$ and $v^s$. Each entry in $u^s$ and $v^s$ will be of the form $(val, pos)$ such that $u_{pos} = val$ and $v_{pos} = val$, respectively, i.e. $pos$ will record the position of a given value in $u$ and $v$. 
We define the entry $(i,j)$ in the outer product $u^s v^s$ as a $(val_u val_v, (pos_u, pos_v))$ such that $u^s_i = (val_u, pos_u)$ and $u^s_j = (val_v, pos_v)$. Comparing the entries on the $val_u val_v$ values, we can assume that we compute the outer product of two sorted vectors.

Assume we have computed the data structure $\mathcal{L}$ implicitly representing the largest $b$ entries in $u^sv^s$, as shown in Lemma~\ref{rank}.
Now we show how to update the summary with the entries in $\mathcal{L}$ in time $O(b + \text{sort}(n))$. We introduce a {\em position total order} on entries such that $(i_1, j_1) < (i_2, j_2)$ iff $i_1n + j_1 < i_2n + j_2$, $i, j \in [n]$. We will keep the entries in $\mathcal{S}$ in the summary sorted according to this order. Assume we can output the $b$ heaviest entries from a given outer product sorted according to the position total order in $\mathcal{L}$. Then in a merge-like scan through $\mathcal{S}$ and $\mathcal{L}$ we update the entries in $\mathcal{S} \cap \mathcal{L}$, remove those from $\mathcal{L}$ and obtain a sorted data structure containing the entries from $\mathcal{S}$ and $\mathcal{L}$ in $O(b)$ steps. The entry of rank $b+1$ in the set $\mathcal{L}\cup \mathcal{S}$, which has size at most $2b$, can be found in $O(b)$ by~\cite{median}. Thus, if the entries in $\mathcal{L}$ are sorted according to the position total order, updating the summary will run in $O(b)$ steps.

We output the $b$ heaviest entries sorted according to the position total order by the following algorithm. Let $\mathcal{L}$ be implicitly given as a sorted column vector $u^s$ and a sorted row vector $v^s$ as described above, and $\ell \leq n$ integer pairs $(q, r_q)$ denoting that in the $q$th row in the outer product $u^sv^s$ the first $r_q > 0$ entries have rank not more than $b$. Clearly, $r_q$ will monotonically decrease as $q$ increases. We start with $q= \ell$, namely the shortest interval, sort the $r_q$ entries according to the position total order. We then decrease $q$ by 1 and sort the next $r_{q-1}$ entries according to the position total order. However, we observe that due to monotonicity $r_{q-1} \geq r_q$ and all elements from $v^s$ appearing in the $q$th row of $u^sv^s$ also appear in the $(q-1)$th row. Thus, we can sort only the new $r_{q-1} - r_q$ elements and then merge the result with the already sorted $r_q$ elements. We continue like this until the elements in each row of the outer product have been sorted. Then we sort the elements in the column vector $u^s$ according to their position, keeping a pointer to the corresponding sorted subinterval of $v^s$ for each entry in $u^s$. From this we build the set $\mathcal{L}$ with entries sorted according the position total order. By setting $r_{\ell + 1} = 0$ the running time for the $\ell$ mergings and sortings  amounts to $\sum_{i=0}^\ell (r_i + \text{Sort}(r_i - r_{i+1}))$. We can bound this sum by $O(b + \text{Sort}(n))$ since $\sum_{i=0}^\ell r_i = b$, $\sum_{i=0}^\ell (r_i - r_{i+1}) \leq n$ and $\sum_{i=0}^{n-1} f(x_i) \le f(\sum_{i=0}^{n-1} x_i)$ for a monotonically growing superlinear function and numbers $x_i$, $i \in [n]$, in its domain. 
\end{proof}

%
%
%
%
%
%
%
%
%

\subsection{Analysis of the approximation guarantee.}

%
%

The only remaining component is how to efficiently answer queries to the summary after processing all outer products. We use a static dictionary with constant look-up time. Observing that the entries are from a universe of size $n^2$, the best known result by Ru\v zi\'c~\cite{pucai_kume} provides a construction in time $O(b \log^2 \log b)$ and space $O(b)$.  Note that $b < n^2$, therefore as a first result we observe that Lemmas~\ref{correctness} and~\ref{rtime_op} immediately yield the following

\begin{lemmx} \label{approx_lemma}
Given $n \times n$-matrices $A,B$ with non-negative real entries, there exists a deterministic algorithm approximating the weight of each entry in the product $C$ of $A$ and $B$ within an additive error of ${\|C\|_{E_1}}/{b}$. The algorithm runs in time $O(nb + n \text{Sort}(n))$ and space $O(b+n)$ in one pass over the input matrices. 
\end{lemmx}

It was first observed by Bose et al.~\cite{bose_et_al} that the {\sc Frequent} algorithm guarantees tighter estimates for items with weight significantly larger than ${N}/{b}$ in a stream of length $N$ and summary of size $b$. Berinde et al.~\cite{berinde_et_al} develop a general framework for the analysis of so called {\em heavy-tolerant} counter based algorithms and show that {\sc Frequent} falls in this class. In order to keep the paper self-contained instead of referring to the results in~\cite{berinde_et_al} we show why our algorithm provides much better guarantees for skewed distributions than suggested by Lemma~\ref{approx_lemma}. Also, the analysis we present is simpler since it does not apply to a wide class of counter based algorithms.

Let us first give some intuition why better approximation guarantees are possible. In order to bound the possible underestimation of an entry in the summary we assume that each decrement of a given entry is applied to at least $b$ distinct entries. Since the total weight of all entries is $\|C\|_{E_1}$ we obtain that the total weight charged to a given entry is bounded by ${\|C\|_{E_1}}/{b}$. If we had only $O(b)$ different entries all having weight about ${\|C\|_{E_1}}/{b}$, then the approximation error given by Lemma~\ref{approx_lemma} is tight. However, assume some entries have weight $\alpha \|C\|_{E_1}$ for $\alpha > {1}/{b}$. This means that the total weight we might have charged when decrementing the weights of groups of at least $b$ distinct entries can be bounded by $\|C\|_{E_1} - (\alpha - {1}/{b})\|C\|_{E_1}$. Iteratively applying the argument to all heavy entries we arrive at improved approximation guarantees. The results in \cite{berinde_et_al,bose_et_al} are based upon the above observation. 

\begin{lemmx}(Bose et al, \cite{bose_et_al}) \label{boseetal}
For an entry $(i,j)$ in $C = AB$ with weight $\alpha \|C\|_{E1}$, $\alpha b > 1$, after termination of {\sc ComputeSummary} it holds $\widehat{C_{ij}} \geq C_{ij} - (1-\alpha){\|C\|_{E1}}/{(b-1)}$ where $\widehat{C_{ij}}$ is the approximation of ${C_{ij}}$ returned by {\sc EstimateEntry}$(i,j)$. 
\end{lemmx} 
\begin{proof}
By Lemma~\ref{key_lemma} we have that the underestimation is at most $\frac{\|C\|_{E_1}}{b}$. As outlined above this implies that we have a total weight of at most $\|C\|_{E_1}(\alpha-\frac{1}{b})$ to charge from, which implies a new upper bound on the underestimation of $\|C\|_{E_1} - \|C\|_{E_1}(\alpha - \frac{1}{b}) = \frac{\|C\|_{E_1}(1-\alpha)}{b} + \frac{\|C\|_{E_1}}{b^2}$.
Successively applying the above reasoning $k$ times we can bound the underestimation to $\sum_{i=1}^k\frac{(1-\alpha)\|C\|_{E_1}}{b^i} + \frac{\|C\|_{E_1}}{b^k}$. Since $b> 1$ for $k \rightarrow \infty$ the claim follows. 
\end{proof}

\begin{lemmx}(Berinde et al, \cite{berinde_et_al}) \label{tail_lemma}
${\|C\|_{E^k1}}/{(b-k)}$ is an upper bound on the underestimation of any ${C_{ij}} $ returned by {\sc EstimateEntry}$(i,j)$ for any $k\leq b$.
\end{lemmx}
\begin{proof}
Assume that $\delta$ is an upper bound on the underestimation returned by {\sc EstimateEntry}. We apply again the above reasoning and since $k < b$ we see that the underestimation is bounded by $\frac{k \delta + \|C\|_{E^k1}}{b}$. We can continue iterating in this way setting $\delta_i = \frac{k \delta_{i-1} + \|C\|_{E^k1}}{b}$ while $\delta_i \leq \delta_{i-1}$ for the underestimation in the $i$th step holds. We either reach a state where no progress is made or, since the underestimation is lower bounded by $\|C\|_{E^k1}$, we have $\delta_i \rightarrow \delta$ as $i \rightarrow \infty$. In either case the claim follows. 
\end{proof}
\\\\
The above lemmas are important since they yield approximation guarantees depending on the residual $k$-norm of the matrix product, thus for skewed matrix products the approximation is much better than the one provided by Lemma~\ref{approx_lemma}.
%
%
%
%
%
%
%
%
%
\\\\
{\bf Sparse recovery.}
The approximation of the matrix product $C=AB$ in \cite{drineas_et_al_mm,pagh_mm,sarlos_mm} is analyzed in terms of the Frobenius norm of the difference of $C$ and the obtained approximation $\widehat{C}$, i.e $\|C-\widehat{C}\|_F$. By simply creating a sparse matrix with all non-zero estimations in the summary we obtain an approximation of $C$: the so called $k$-sparse recovery of a frequency vector $\bf{f}$ aims at finding a vector $\bf{\widehat{f}}$ with at most $k$ non-zero entries such that the $p$-norm $\|{\bf f} - {\bf \widehat{f}}\|_p$ is minimized. 

As shown by Berinde et al.~\cite{berinde_et_al} the class of heavy-tolerant counter algorithms yields the best known bounds for the sparse recovery in the $p$-norm. The following Theorem 1 follows from Lemma~\ref{rtime_op} and their main result.

\begin{thmx}
Let $A, B$ be nonnegative $n \times n$ real matrices and $C = AB$ their product. There exists a one-pass approximation deterministic algorithm returning a matrix $\widehat{C}$ such that $\|C - \widehat{C}\|_{Ep} \leq (1+\varepsilon)^{\frac{1}{p}}({\varepsilon}/{k})^{1-\frac{1}{p}}\|C\|_{E^k1}$. The algorithm runs in time $O(n\cdot \text{Sort}(n) + (n{k})/{\varepsilon})$ and uses space $O(n+{k}/{\varepsilon})$ for any $0 < \varepsilon < 1$ and $k \geq 1$. 
\end{thmx}
\begin{proof}
Let us write the entries of $C$ in ascending order according to their weight as $C_0, C_1,\ldots,C_{n^2-1}$, thus $\|C\|_{E^kp} = (\sum_{\ell = k}^{n^2-1}C_\ell^p)^\frac{1}{p}$. Observe that by setting $b = k + \frac{k}{\varepsilon}$ the bounds given by Lemma~\ref{tail_lemma} yield an underestimation for every entry of at most $\frac{\varepsilon\|C\|_{E^k1}}{k}$. Thus, we can upper bound $(\|C - \widehat{C}\|_{Ep})^p$ as $k (\frac{\varepsilon\|C\|_{E^k1}}{k})^p + \sum_{\ell = k}^{n^2-1} |C_\ell - \widehat{C_\ell}|(\frac{\varepsilon}{k})^{p-1}(\|C\|_{E^k1})^{p-1}$. The last term is at most $(\frac{\varepsilon}{k})^{p-1}(\|C\|_{E^k1})^p$, thus we can upper bound $(\|C - \widehat{C}\|^p)^p$ to $(1+\varepsilon)\frac{\varepsilon}{k}(\|C\|_{E^k1})^p$. The time and space complexity follow directly from Lemma~\ref{rtime_op}.~\end{proof}
\\\\
Clearly, for $k/\varepsilon = o(n^{2})$ the algorithm runs in subcubic time and subquadratic memory. In the next paragraph we show that for skewed output matrices {\sc EstimateEntry} can provably detect the most significant entries even for modest summary sizes.
%
%
%
%
%
\\\\
{\bf Zipfian distributions.}
As discussed in Section~\ref{skew_appl} the assumption that the entries in the product adhere to a Zipfian distribution is justified. The results stated below not only give a better understanding of the approximation yielded by the algorithm, but also allow direct comparison to related work. 

\begin{lemmx}(Berinde et al, \cite{berinde_et_al})
If the entries weights in the product matrix follow a Zipfian distribution with parameter $z>1$, then {\sc EstimateEntry} with a summary of size $b$ 
\begin{enumerate}
\item approximates the weight of all entries with rank $i \leq b$ with additive error of 
$(1 - \frac{1}{\zeta(z)i^z})\frac{\|C\|_{E1}}{b-1}$.
\item estimates the weight of all entries with additive error of $\varepsilon \|C\|_{E1}$ for $b =  O((\frac{1}{\varepsilon})^{\frac{1}{z}})$.
\item returns the largest $k$ entries in the matrix product for $b = O(k)$.
\item returns the largest $k$ entries in a correct order for $b=\Omega(k(\frac{k}{z})^{\frac{1}{z}})$.
\end{enumerate}
\end{lemmx}
\begin{proof}
The first statement is a direct application of Lemma~\ref{boseetal}. For 2. we choose $k=b/2$ in Lemma~\ref{tail_lemma}. The error is then bounded by $\frac{2\|C\|_{E^{b/2}1}}{b}$. By plugging in the bounds on the tail of the Zipfian distribution we obtain that the error is bounded by $O(\frac{\|C\|_{E1}}{b^z})$. This implies that there exists $b = c_bk$ such that all entries of weight $\frac{\|C\|_{E1}}{c_bk^z}$, for some constant $c_b$, will be in the summary. This shows 3. Since the entry of rank $k$ has weight $\frac{\|C\|_{E1}}{\zeta(z)k^z}$ and we assume that $\zeta(z)$ is constant.
For the last statement we observe that the function $f(x) = \frac{1}{x^z} - \frac{1}{(x+1)^{z}}$ is monotonically decreasing for $x>0$ and $z>1$. Thus we need an error of at most $2\delta < \|C\|_{E1}(\frac{1}{k^z} - \frac{1}{(k+1)^{z}})$. With some algebra we can bound the difference to $O(\frac{z\|C\|_{E1}}{k(k+1)^z})$, thus by $2.$ the claim follows. 
\end{proof}

%
%
%
%
%
%
%
%

\subsection{Comparison to previous work.}

The randomized algorithm by Cohen and Lewis~\cite{cohen_lewis_mm} for computing the product of nonnegative matrices yields an unbiased estimator of each entry and a concentration around the expected entry weight with high probability. However, their algorithm requires a random walk in a bipartite graph of size $\Theta(n^2)$ space and is thus not space efficient. It is difficult to compare the bounds returned by {\sc EstimateEntry} to the bounds obtained in~\cite{drineas_et_al_mm,sarlos_mm}, but it is natural to compare the guarantee of our estimates to the ones shown by Pagh~\cite{pagh_mm}.

\begin{figure}[ht]
\begin{center}

\includegraphics[scale=0.55]{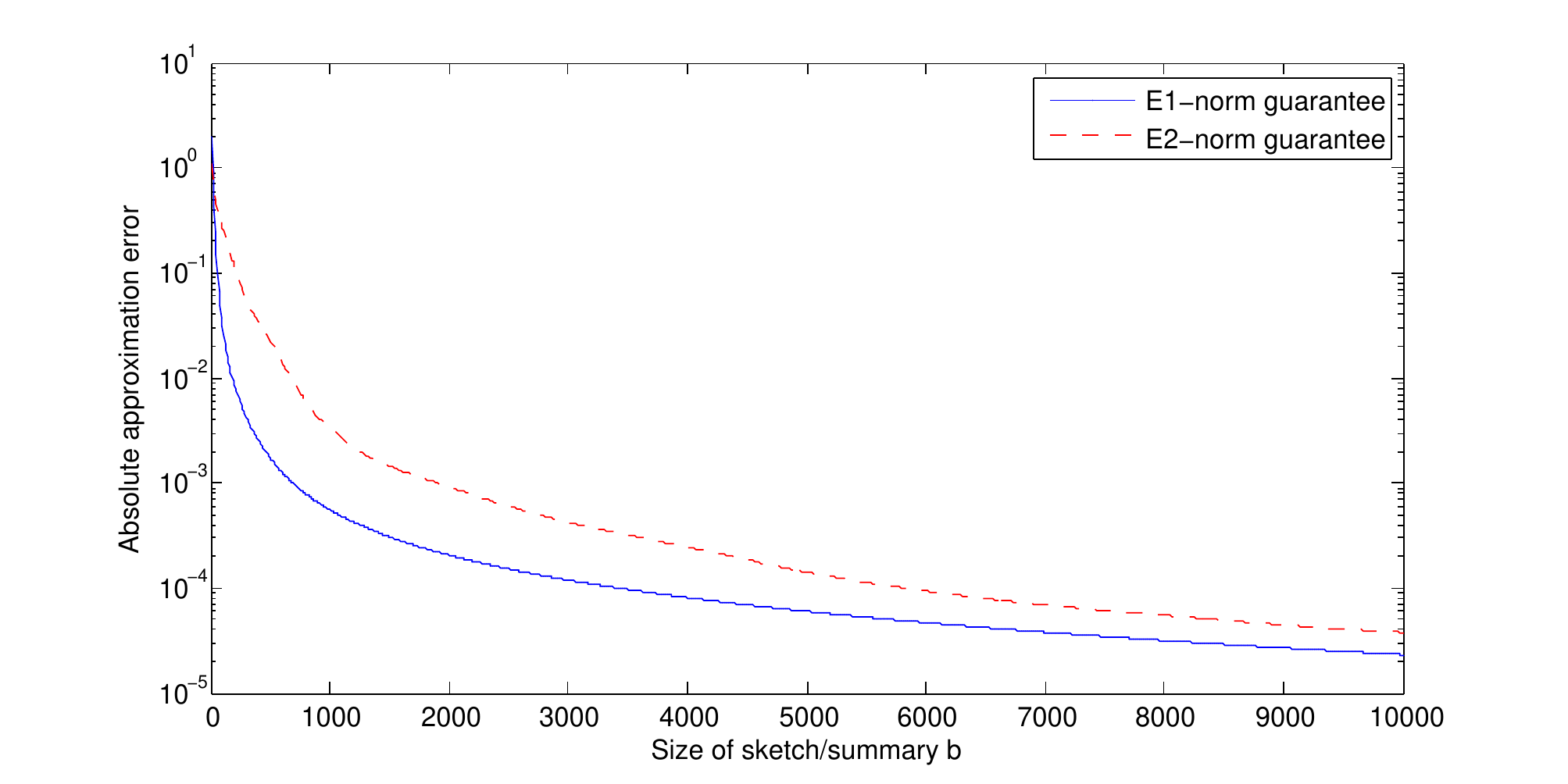}

\end{center}
\caption{We compare the theoretical error guarantee achievable by Pagh's~\cite{pagh_mm} and our algorithm for the estimation of individual entries for sketch of size $b$ for the dataset accidents. }\label{fig:pagh_vs_kk}
\end{figure}

The approximation error of the matrix estimation $\widehat{C}$  in~\cite{pagh_mm}, $\|C - \widehat{C}\|_F$, is bounded by ${(n\|C\|_F)}/{\sqrt{b}}$ with high probability. The running time is $O(n^2\log n + b\log b \log n)$ and space usage is $O(n + b\log n)$. Our deterministic algorithm achieves an error guarantee of $(1+\varepsilon)({\varepsilon}/{k})^{1-\frac{1}{p}}\|C\|_{E^k1}$ for the approximation $\|C-\widehat{C}\|_{Ep}$ for any $p>0$. For a direct comparison we set $p=2$, $k=0$ and $b = \lceil {1}/{\varepsilon}\rceil$ and obtain an approximation error of ${\|C\|_{E1}}/{\sqrt{b}}$ which is at most ${(n\|C\|_F)}/{\sqrt{b}}$ by Cauchy-Schwarz inequality. The time and space complexity of our algorithm is a polylogarithmic factor better. Note also that the approximation guarantee does not depend on the dimension $n$ as in~\cite{pagh_mm}.

For individual entries we achieve an error bounded by $\text{min}_{k\in [b]} {\|C\|_{E^k1}}/{(b-k)}$ while~\cite{pagh_mm} shows that the error of the obtained estimates is bounded by ${\|C\|_{E^{b/\kappa}1}}/{\sqrt{b}}$ for a suitably chosen constant $\kappa > 2$. 

Assuming Zipfian distribution with $z>1$ the approximation error of the Frobenius norm of the matrix product in~\cite{pagh_mm} is bounded by $O(nb^{-z}\|C\|_{E1})$ with high probability. By setting $k=0$ our deterministic algorithm achieves $O({\|C\|_{E1}}/{\sqrt{b}})$ for the Frobenius norm approximation error. 
For an $\varepsilon\|C\|_{E1}$-approximation of individual entries both~\cite{pagh_mm} and our algorithm need a data structure, a sketch or a summary, of size $O(({1}/{\varepsilon})^{\frac{1}{z}})$ but~\cite{pagh_mm} needs to run $O(\log n)$ copies of the algorithm in parallel in order to guarantee that the estimates are correct with high probability. Figure~\ref{fig:pagh_vs_kk} plots the additive approximation achieved for different summary sizes for the accidents dataset. The dashed line gives the theoretical guarantees on the approximation error shown in~\cite{pagh_mm}, depending on the entrywise 2-norm of the matrix product, and the solid line the approximation we achieve (depending on the entrywise 1-norm).   In order to achieve guarantees with high probability Pagh's algorithm needs to run $O(\log n)$ copies of the algorithm in parallel, thus increasing the time and space complexity by a factor of $O(\log n)$.  
However, Pagh's algorithm achieves better bounds for lighter skew when $1/2 < z < 1$ and more important it is not restricted to nonnegative input matrices.

%
%
%
%
%
%
%

%
%
%

%
%
%
%
%
%
%

\section{An algorithm for arbitrary real-valued matrices}

In this section we show how to efficiently extend the deterministic streaming algorithm sketched in~\cite{hot,muthu_survey} to matrix multiplication. The algorithm in~\cite{hot,muthu_survey} works for streams in the non-strict turnstile model where updates are of the form $(i, v)$ for an item $i$ and $v \in \mathbb{R}$. 

Let us first consider the generalized Majority algorithm for the non-strict turnstile model. We want to find an item whose absolute total weight is more than half of the absolute sum of total weights of all other items in the stream. Recall that in the non-strict turnstile model an item is allowed to have negative weight after processing the stream. We adapt the key observation in~\cite{hot} to obtain

\begin{lemmx} \label{majority_bits}
A majority item in a real weighted stream over a domain of size $m$ in the non-strict turnstile model can be determined in two passes using $\lceil \log m \rceil + 1$ counters and processing time per item arrival.
\end{lemmx}
\begin{proof}
Assume the items are integers over some finite domain $[m]$. We form $\lceil \log m \rceil$ groups $g_k$, $0 \leq k \leq \lceil \log m \rceil -1$, one for each bit of the binary representation for each item $i \in [m]$. For each arriving item $i$ we compute its $k$th bit and if it is 1 we add $i$'s weight to the counter $c_k$ of the $k$th group. We also maintain the total weight of the stream seen so far in a global variable $w$. After the stream has been processed we first determine the sign of $w$.  Each item either increases the weight of the $k$th group or does not change it. Since the set of items is divided in two disjoint subsets for the $k$th group depending on the $k$th bit of each item, given $w$ and $c_k$ we can also determine the total weight of items whose $k$th bit is equal to 0, call it $w^k_0$. First observe that an item with absolute weight more than $|w|/2$ must have the same sign as $w$. Thus, if the weights in the two groups for the $k$th bit have different signs we determine the $k$th bit of the potential majority item. Otherwise, if the 0-group and the 1-group have weight with equal sign, the majority element must be in the subset with larger absolute weight. We show this by contradiction. Set $w^k_1:=c_k$. Assume w.l.o.g. that both groups have positive weights  $w^k_0>0$, $w^k_1>0$, $w^k_0 \geq w^k_1$ and the $k$th bit of the majority item is 1. Let the weight of the majority item $m$ be $w_m$. We have $w^k_1 \geq w_m - w^k_{-,1}$ where $w^k_{-,1}$ is the absolute value of the total contribution of entries with negative total weight and $k$th bit 1. The last inequality follows from the fact that $w_m$ is positive and the sum of weights of items with positive total weight is lower bounded by $w_m$. But this implies $w_m \leq w^k_1 + w^k_{-,1} \leq w^k_0 + w^k_{-,1}$ which contradicts that $m$ is the majority item.  Therefore it must hold $w^k_0 < w^k_1$.

After the first pass we can construct a candidate majority item from the $\lceil \log m \rceil$ bits found. A second pass over the stream will reveal whether indeed the candidate is a majority item.  
\end{proof}
\\\\
We generalize the above algorithm to finding a majority entry in a matrix product. Generalizing again builds upon the column row method.
\begin{lemmx} \label{maj_matrix}
A Majority entry in a matrix product can be computed in time $O(n^2\log n)$, linear space and two passes over the input. 
\end{lemmx}
\begin{proof}
Let us assume $n = 2^\ell$ for some integer $\ell > 0$. We number the $n^2$ entries of the matrix product as $0,1,\dots, n^2-1$ such that the entry in the position $(i,j)$ is assigned a number $i2^\ell + j$, $0 \leq i,j \leq n-1$. Now observe that each entry of the matrix product consists of $2\ell$ bits and the term $j$ determines only the least significant $\ell$ bits while the term $i2^\ell$ determines the most significant $\ell$ bits. In the sequence $0,1,\ldots,n^2-1$ the elements having 1 in the $k$th position, $0 \leq k \leq 2\ell -1$, are the ones in positions $\{2^{k}+i2^{k+1},\dots, (i+1)2^{k+1} - 1\}$, $0 \leq i \leq  2^{2\ell-k-1}$. Thus, given a column vector $a$ of $A$ and a row vector $b$ of $B$ the entries in the outer product $ab$ with the $k$th bit equal to 1 are uniquely determined by the position of the contribution from either $a$ or $b$. More concretely, for $k < \ell$ and $k \geq \ell$ these are the entries in positions $\{2^{k}+i2^{k+1},\dots, (i+1)2^{k+1} - 1\}$, $0 \leq i \leq 2^{\ell-k-1}$, in $a$ and $b$, respectively.
  
Thus, to find the total contribution of entries weights with a bit set to 1 in a given position we need to simply consider only the entries in the corresponding interval and by the distributive law for the ring of real numbers we can compute the total contribution to the group for the $i$th 1-bit in $O(n)$ steps. Writing the matrix product as a sum of $n$ outer products and applying to the resulting stream the Majority algorithm from Lemma~\ref{majority_bits} with the above outlined modification yields the claimed running time. 
\end{proof}

\begin{figure}[h!]
\Kw{\sc ComputeGroups}
\algsetup{indent=2em}
\begin{algorithmic}[1]
\REQUIRE matrices $A, B \in \mathbb{R}^{n\times n}$, an array of primes $P$, integer $b$ 
\STATE Create data structures $Q$ of size $X \times Y$ and $G$ of size $X \times Y \times L$ for $X = c_xb \log n, Y = c_yb\log n\log b, L = 2\ell$ for suitably chosen constants $c_x$ and $c_y$ and $\ell = \lceil \log n \rceil$
\FOR{$i \in [n]$}
\STATE $a= A_{*,i}$, $b = B_{i,*}$
\FOR{$j := 0 $ to $|P|$}
\STATE prime $p := P[j]$
\STATE $p_a = \sum_{t=0}a_tx^{t\cdot n \text{ mod } p}$
\STATE $p_b = \sum_{t=0} b_t x^{t \text{ mod } p}$
\STATE $p_{ab} = FFT(p_a, p_b)$
\STATE $\overline{p_{ab}} = p_{ab} \text{ mod } x^p + p_{ab} \text{ div } x^p$
\FOR{$m \in [p]$}
\STATE $Q[j][m] := Q[j][m] + c_m$ for $c_mx^m \in \overline{p_{ab}}$
\ENDFOR
\FOR{ $k \in [2\ell]$}
\IF{$k \leq \ell$}
\STATE set to 0 the entries $a_t$ for $t \in \{i2^{k+1},\dots, (i+1)2^{k+1} -2^k - 1\}$, $0 \leq i \leq  2^{2\ell-k-1}$
\ELSE
\STATE set to 0 the entries $b_t$ for $t \in \{i2^{k+1},\dots, (i+1)2^{k+1} -2^k - 1\}$, $0 \leq i \leq  2^{2\ell-k-1}$
\ENDIF
\STATE $p_a = \sum_{t=0}a_tx^{t\cdot n \text{ mod } p}$
\STATE $p_b = \sum_{t=0} b_t x^{t \text{ mod } p}$
\STATE $p_{ab} = FFT(p_a, p_b)$
\STATE $\overline{p_{ab}} = p_{ab} \text{ mod } x^p + p_{ab} \text{ div } x^p$
\FOR{$m \in [p]$}
\STATE $G[j][m][k] = G[j][m][k] + c_m$ for $c_mx^m \in \overline{p_{ab}}$ 
\ENDFOR
\ENDFOR
\ENDFOR
\ENDFOR
\end{algorithmic}
\bigskip
\Kw{\sc CheckCandidates}
\algsetup{indent=2em}
\begin{algorithmic}[1]
\REQUIRE Data structures $G$ and $Q$, array of prime numbers $P$
\STATE Initialize a $X \times Y$ matrix $K$ with $2\ell$-bitstrings for $X = c_xb \log b, Y = c_yb\log n\log b$
\FOR{$j := 0 $ to $|P|$}
\STATE prime $p := P[j]$
\FOR{$m \in [p]$}
\FOR{$k \in [2\ell]$}
\STATE Set $k$th bit of $K[j][m]$ depending on the values $Q[j][m]$ and $G[j][m][k]$
\ENDFOR
\ENDFOR
\ENDFOR
\FOR{$\kappa \in K$}
\STATE set $i$ from the first $\ell$ bits in $\kappa$
\STATE set $j$ from the last $\ell$ bits in $\kappa$
\STATE Check the value of $AB(i, j)$
\ENDFOR
\end{algorithmic}
\caption{The function {\sc ComputeGroups} distributes the entries to different groups depending on their {\em mod} value for a prime number $p$. In each group we update $2\ell + 1$ subgroups as outlined in Lemma~\ref{maj_matrix}. In {\sc CheckCandidates} we construct a candidate entry for each group and find its exact weight in the matrix product.} \label{alg2}
\end{figure}

Theorem 14 in~\cite{muthu_survey} presents a generalization of the Majority algorithm for non-strict turnstile data streams where at most $k$ items have a value significantly different from 0 after processing the stream. We combine the approach with the technique presented by Pagh~\cite{pagh_mm} to obtain our main theorem for the multiplication of arbitrary real valued matrices. The proof of the theorem analyzes the complexity of the algorithm in Figure~\ref{alg2}.

\begin{thmx}
Let $A, B$ be real $n \times n$ matrices and $C = AB$ their product.  If the absolute weight of each of the $b$ entries with largest absolute weight is bigger than $\|C\|_{E^b1}$, then there exists a deterministic algorithm computing the $b$ heaviest entries exactly in time $O(n^2 + nb^2 \log^3 n \log^2 b)$ and space $O(n + b^2 \log^3 n \log b)$ in two passes over the input. 
\end{thmx}
\begin{proof}
We first describe the approach presented in~\cite{muthu_survey} for data streams over the non-strict turnstile model and then adapt it to matrix multiplication using Lemma~\ref{maj_matrix} and the approach from~\cite{pagh_mm}.

Let $P = \{p_1, p_2, \ldots, p_x\}$ be a set of consecutive prime numbers such that $b < p_1 < p_2 < \ldots < p_x$ and  $x=(b-1)\log_bN + 1$ where $N$ is the cardinality of the set of items in the input stream. Assume that at most $b$ items will have weights different from 0. For each prime number $p$ we create $p$ groups labelled as $0,1,\ldots,p-1$. Now, for a given $p$, each incoming item $i$ is distributed to exactly one of the $p$ groups by computing $i \mbox{ mod } p$. The crucial observation is that for any subset of $b$ non-zero items each one of them will be isolated in at least one group. This is true because two different items can share at most $\log_bN$ groups for different prime numbers. Otherwise the absolute value of their difference modulo $p$ is 0 for $\log_bN + 1$ distinct primes larger than $b$ which contradicts the fact that all items are smaller than $N$. Thus, for a given non-zero item $i$  each of the other $b-1$ non-zero items ``locks" at most $\log_bN$ groups for different primes. Since we have $(b-1)\log_bN + 1$ distinct primes there must exist a prime $p$ such that $i \mbox{ mod } p = r (\mbox{ mod } p)$ implies $j \text{ mod } p \neq r (\text{ mod } p)$ for all $j \in B, j \neq i$. Thus running the Majority algorithm from Lemma~\ref{majority_bits} for each group and checking the results in a second pass will reveal the $b$ non-zero items. Clearly, the algorithm is resilient to noise in the zero-valued items. 

We now adjust the approach to matrix multiplication with sparse output. Figure~\ref{alg2} presents a pseudocode description of the algorithm. Given $n\times n$ real input matrices $A$ and $B$ and a parameter $b$ we assume we have found the $(b-1)\log_bn^2 + 1$ consecutive prime numbers larger than $b$, denote them as $\mathcal{P}$. As in the algorithm for nonnegative input matrices we iterate over $n$ outer products and consider each of them as $n^2$ distinct item arrivals. We iterate over primes $p \in \mathcal{P}$. We concentrate how {\sc ComputeGroups} works for a fixed prime $p \in \mathcal{P}$ and an outer product $ab$ for $a = A_{*,i}, b= B_{i, *}$, $i \in [n]$. We want to distribute each of the $(i, j)$ entries in $ab$ in $p$ groups depending on $i, j$ and $p$ and then run the  Majority algorithm from Lemma~\ref{majority_bits} in each group. We write each of the $n^2$ entries as $i \cdot n + j$, for $i,j \in [n]$. Since for any two $n_1, n_2 \in \mathbb{N}$ we have $(n_1 + n_2) \text{ mod } p = (n_1  \text{ mod } p +  n_2 \text{ mod } p) \text{ mod } p$ we want to compute $\sum_{i,j} a_i b_j$ such that $(i\cdot n \text{ mod } p  + j \text{ mod } p) \text{ mod } p  = \ell$ for each $\ell \in [p]$. The na\"ive solution would explicitly compute the $\text{mod } p$ value of each of the $n^2$ entries and update the corresponding group.  As observed by Pagh~\cite{pagh_mm}, however, one can treat the two vectors as polynomials of degree $p-1$: $p_a = \sum_{i=0}^{n-1} a_ix^{i\cdot n \text{ mod } p}$ and $p_b = \sum_{j=0}^{n-1} b_jx^{j \text{ mod } p}$, lines 6--7. In line 8 we multiply the two polynomials in time $O(p\log p)$ using FFT \footnote{Note, that the standard FFT algorithm assumes the highest order is a power of 2, thus we have to add $<p$ high order terms to $p_a$ and $p_b$ with coefficients equal to 0 to obtain the product $p_ap_b$.}. The product is $p_ap_b = \sum_{k=0}^{2(p-1)} c_kx^k$ where $c_k = \sum_{i,j} a_ib_{j}$ such that $(i \cdot n \text{ mod } p) + (j \text{ mod } p) = k$, thus we only have to add up the entries with the same exponents modulo $p$: ($p_ap_b \text{ mod } x^p + p_ap_b \text{ div } x^p$). In line~11 we update the global counter from Lemma~\ref{majority_bits} recording the total contribution of entries weights in the stream and store it in a matrix $Q$ in a corresponding cell.  Next, in lines 12--22, we repeat essentially the same procedure for each of $2\ell$ bits of each entry: we  nullify the respective entries in either $a$ or $b$, those with a 0 in the $k$th bit, as outlined in Lemma~\ref{maj_matrix}. Then we multiply the updated polynomials and add to the counter of the corresponding (sub)groups the newly computed coefficients, stored in a data structure $G$.  

After the first pass over the input matrices we have to construct a majority candidate for each of the $p$ subgroups of a prime $p \in \mathcal{P}$ and check whether it is indeed different from 0 in a second pass over the input matrices. In {\sc CheckCandidates} we do this by using the previously computed data structures $G$ and $Q$ from which we extract the required information as outlined in Lemma~\ref{maj_matrix}.

Correctness of the algorithm follows from the above discussion and the previous lemmas.

For the complexity analysis we need to know the cardinality of $\mathcal{P}$ and the order of the largest prime in $\mathcal{P}$. In the following we assume $b>2$. We have $|\mathcal{P}| = O(b\log_b n) = O(b \log n)$ and from the prime number theorem we observe that for the largest $p \in \mathcal{P}$ we have $p = O(b \log_b n(\log b + \log \log_b n)) = O (b \log b\log n)$. Therefore, for a given outer product we iterate over $O(b \log n)$ primes and in each iteration $O(\log n)$ times we multiply using FFT two polynomials of degree $O(b \log b\log n)$. This yields a total running time of 
$O(b^2\log^3 n \log^2 b)$. For the data structures $G$ and $Q$ we need a total space of $O(b^2\log^3 n \log b)$. {\sc CheckCandidates} needs time of $O(nb^2\log^3n \log b)$ to find the $b$ nonzero entries in the product. 
\end{proof}
\\\\
The above theorem immediately implies the following result for sparse matrix products: 

\begin{corollary}
Let $A, B$ be real $n \times n$ matrices and $C = AB$ their product. If $C$ has at most $b$ nonzero entries then there exists a deterministic algorithm computing $C$ exactly in time $O(n^2 + nb^2 \log^3 n \log^2 b)$ and space $O(n + b^2 \log^3 n \log b)$ in two passes over the input. 
\end{corollary}
%
%
%
%
%
%
%
\subsection{Zipfian distribution.}

For the case when the absolute values of the entries in the outer product adhere to Zipfian distribution with parameter $z>1$ we obtain the following 

\begin{thmx}
Let the absolute values of the entries weights in a matrix product adhere to Zipfian distribution. Then for user-defined $s>0$ and $k>0$ there exists a deterministic algorithm detecting the $ks$ heaviest entries in the product in time $O(s(n^2 + nk^{\frac{2z}{z-1}}\log^3 n \log^2 k))$ and space $O(n + ks + k^{\frac{2z}{z-1}}\log^3 n \log k)$ in $2s$ passes over the input matrices. 
\end{thmx} 

\begin{proof}
We run {\sc ComputeGroups} with suitably chosen parameters such that an entry with absolute weight of at least $f:=\frac{\|C\|_{E_1}}{\zeta(z)k^z}$ lands in at least one group such that the contribution from other entries is less than $f$. Thus, we have to isolate the $x$ heaviest entries in different groups such that $\frac{1}{k^z} \geq \sum_{i=x+1}^{n^2}\frac{1}{i^z}$. Since $\sum_{i=x+1}^{n^2}\frac{1}{i^z} \leq cx^{1-z}$ for some constant $c$, we need $\frac{1}{k^z} \geq cx^{1-z} \Leftrightarrow x \geq ck^{\frac{z}{z-1}}$. In a second pass we run {\sc CheckCandidates} and find the $k$ heaviest entries. 

A natural generalization for detecting the heaviest entries in $s$ passes for a user-defined $s$ works as follows. In one pass we will run {\sc ComputeGroups} and in the subsequent pass {\sc CheckCandidates}. Assume after the $2i$th pass, $1 \leq i \leq s-1$, we have found the $ik$ entries with largest absolute weight and they are stored in a set $\mathcal{S}$. Then, in the $(2i+1)$th pass we run {\sc ComputeGroups} but at the end we subtract the weights of all entries in $\mathcal{S}$ from the counters in the corresponding groups in order to guarantee that the $ik$ heaviest entries are not considered. After finding the new $k$ heaviest entries we add them to $\mathcal{S}$.  
\end{proof}

\subsection{Comparison to previous work.} The algorithm seems to be only of theoretical interest. Note that its complexity is much worse than the one achieved by Pagh's randomized one-pass algorithm \cite{pagh_mm}: the $b$ non-zero entries can be found in time $O(n^2 + nb \log n)$ and space $O(n + b\log n)$ with error probability $O(1/\text{poly}(n))$. Nevertheless since the best known space lower bound for finding $b$ non-zero elements by a deterministic algorithm is $O(b\log n)$ there seems to be potential for improvement. For example Ganguly and Majumder~\cite{detkset} present improvements of the deterministic algorithm from~\cite{muthu_survey} but their techniques are not suitable for our problem.

To the best of our knowledge this is the first deterministic algorithm for computing matrix products in time $O(n^{2 + \varepsilon})$ for the case when the product contains at most $O(\sqrt{n})$ non-zero entries. The algorithm by Iwen and Spencer achieves this for an arguably more interesting class of matrix products, namely those with  $n^{\beta}$, $\beta \leq 0.29462$, nonzero entries in each row, but the algorithm relies on fast rectangular matrix multiplication and its simple version runs in time $O(n^{2+\beta})$.  
\\\\
{\bf Acknowledgements.} I would like to thank my supervisor Rasmus Pagh for his continuous support and many valuable comments and suggestions.

\bibliographystyle{article}

\begin{thebibliography}{1}

\bibitem{ams}
N. Alon, Y. Matias, and M. Szegedy. 
\newblock The space complexity of approximating the frequency
moments. 
\newblock {\em J. Comput. Syst. Sci}, 58(1):137--147, 1999.

\bibitem{amossen_pagh_mm}
R. R. Amossen and R. Pagh. 
\newblock Faster join-projects and sparse matrix multiplications.
\newblock {\em ICDT 2009}, 121--126.


\bibitem{berinde_et_al}
R. Berinde, P.~Indyk, G. Cormode, M. J. Strauss: 
\newblock Space-optimal heavy hitters with strong error bounds. 
\newblock {\em ACM Trans. Database Syst.} 35(4): 26 (2010)

\bibitem{median}
M. Blum, R. W. Floyd, V. R. Pratt, R. L. Rivest, R. E. Tarjan.
\newblock Time Bounds for Selection.
\newblock {\em J. Comput. Syst. Sci. 7(4)}: 448--461 (1973)

\bibitem{bose_et_al}
P. Bose, E. Kranakis, P. Morin, Y. Tang. 
\newblock Bounds for Frequency Estimation of Packet Streams. 
\newblock {\em SIROCCO 2003}: 33--42

\bibitem{mjrty}
R. Boyer and S. Moore
\newblock A Fast Majority Vote Algorithm
\newblock {\em U. Texas Tech report}, 1982

\bibitem{brin_et_al}
S. Brin, R. Motwani, C. Silverstein.
\newblock Beyond Market Baskets: Generalizing Association Rules to Correlations.
\newblock {\em SIGMOD 1997}: 265--276


\bibitem{mm_history}
P. Burgisser, M. Clausen, and M. A. Shokrollahi. 
\newblock Algebraic complexity theory. 
\newblock {\em Springer-Verlag}, 1997

\bibitem{bisam}
A. Campagna and R. Pagh.
\newblock Finding associations and computing similarity via biased pair sampling.
\newblock {\em Knowl. Inf. Syst.} 31(3): 505--526 (2012)

\bibitem{count_sketch}
M. Charikar, K. Chen, and M. Farach-Colton. 
\newblock Finding frequent items in data streams. 
\newblock {\em Theor. Comput. Sci}, 312(1):3--15, 2004

\bibitem{cohen_et_al}
E. Cohen, M. Datar, S. Fujiwara, A. Gionis, P. Indyk, R. Motwani, J. D. Ullman, C. Yang.
\newblock Finding Interesting Associations without Support Pruning.
\newblock {\em IEEE Trans. Knowl. Data Eng.} 13(1): 64--78 (2001)

\bibitem{cohen_lewis_mm}
E. Cohen and D. D. Lewis. 
\newblock Approximating matrix multiplication for pattern recognition tasks. 
\newblock {\em Journal of Algorithms}, 30(2):211--252, 1999


\bibitem{copwin_mm}
D. Coppersmith and S. Winograd. 
\newblock Matrix multiplication via arithmetic progressions.
\newblock {\em Journal of Symbolic Computation}, 9(3):251--280, 1990

\bibitem{hot}
G. Cormode and S. Muthukrishnan. 
\newblock What's hot and what's not: Tracking most frequent items dynamically. 
\newblock {\em ACM Transactions on Database Systems}, 30(1):249--278, 2005. 

\bibitem{demaine_et_al}
E. D. Demaine, A. L\'opez-Ortiz, J. I. Munro. 
\newblock Frequency Estimation of Internet Packet Streams with Limited Space.
\newblock {\em ESA 2002}: 348--360

\bibitem{mcl}
S. V. Dongen. 
\newblock Graph Clustering by Flow Simulation. 
\newblock {\em PhD thesis}, University of Utrecht, 2000

\bibitem{drineas_et_al_mm}
P. Drineas, R. Kannan, and M. W. Mahoney. 
\newblock Fast Monte Carlo algorithms for matrices I: Approximating matrix multiplication. 
\newblock {\em SIAM Journal on Computing}, 36(1):132--157, 2006

\bibitem{selection_XY}
G. N. Frederickson, D. B. Johnson.
\newblock The Complexity of Selection and Ranking in X+Y and Matrices with Sorted Columns.
\newblock {\em J. Comput. Syst. Sci. 24(2)}: 197--208 (1982)

\bibitem{detkset}
S. Ganguly and A. Majumder. 
\newblock Deterministic $k$-set structure. 
\newblock {\em PODS 2006}: 280--289

\bibitem{best_sort}
Y. Han. 
\newblock Deterministic sorting in $O(n \log \log n)$ time and linear space.
\newblock {\em J. Algorithms 50(1)}: 96--105 (2004)

\bibitem{dm_book}
J. Han, M. Kamber. 
\newblock Data Mining: Concepts and Techniques 
\newblock {\em Morgan Kaufmann} 2000

\bibitem{iwen_spencer_mm}
M. A. Iwen and C. V. Spencer. 
\newblock A note on compressed sensing and the complexity of matrix multiplication. 
\newblock {\em Inf. Process. Lett}, 109(10):468--471, 2009

\bibitem{karp_et_al}
R. M. Karp, S. Shenker, C. H. Papadimitriou. 
\newblock A simple algorithm for finding frequent elements in streams and bags. 
\newblock {\em ACM Trans. Database Syst. 28}: 51--55 (2003)

\bibitem{lingas_mm}
A. Lingas. 
\newblock A fast output-sensitive algorithm for boolean matrix multiplication.  
\newblock {\em ESA 2009}, 408--419.


\bibitem{misra_gries}
J. Misra, D. Gries: 
\newblock Finding Repeated Elements. 
\newblock {\em Sci. Comput. Program.} 2(2): 143--152 (1982)

\bibitem{muthu_survey}
S. Muthukrishnan. 
\newblock Data Streams: Algorithms and Applications. 
\newblock {\em Foundations and Trends in Theoretical Computer Science}, Vol. 1, Issue 2, 2005

\bibitem{pagh_mm}
R. Pagh.
\newblock Compressed Matrix Multiplication.
\newblock {\em Proceedings of ACM Innovations in Theoretical Computer Science (ITCS)}, 2012

\bibitem{pucai_kume}
M. Ru\v zi\'c. 
\newblock Constructing Efficient Dictionaries in Close to Sorting Time. 
\newblock {\em ICALP (1) 2008}: 84--95

\bibitem{sarlos_mm}
T. Sarl\'os. 
\newblock Improved Approximation Algorithms for Large Matrices via Random Projections. 
\newblock {\em FOCS~2006}: 143--152

\bibitem{outer_prod_mm}
C.-P. Schnorr, C. R. Subramanian.
\newblock Almost Optimal (on the average) Combinatorial Algorithms for Boolean Matrix Product Witnesses, Computing the Diameter. 
\newblock {\em RANDOM 1998}: 218--231

\bibitem{stothers_mm}
A. J. Stothers. 
\newblock On the complexity of matrix multiplication. 
\newblock {\em Ph.D. thesis}, University of Edinburgh, 2010

\bibitem{strassen_mm}
V. Strassen. 
\newblock Gaussian Elimination is not Optimal. 
\newblock {\em Numer. Math.} 13, 354--356, 1969

\bibitem{yuster_zwick_mm} 
R. Yuster and U. Zwick. 
\newblock Fast sparse matrix multiplication. 
\newblock {\em ACM Transactions on Algorithms}, 1(1):2--13, 2005.

\bibitem{virgi_mm}
V. Vassilevska Williams.
\newblock  Multiplying matrices faster than Coppersmith-Winograd.
\newblock {\em STOC 2012}, 887--898 

\end{thebibliography}

\end{document}